\documentclass[journal,12pt,onecolumn]{IEEEtran}

\usepackage{amsmath,amsfonts}

\usepackage{amssymb,verbatim,amsfonts,amsmath,amsthm}
\usepackage{geometry}
\usepackage{cite}
\usepackage[linktocpage = true]{hyperref}
\usepackage{bm}

\usepackage{hyperref}
\usepackage{seqsplit}
\usepackage{longtable}
\usepackage{booktabs}
\usepackage[justification=centering]{caption}
\usepackage{graphicx}
\usepackage[figuresright]{rotating}
\usepackage{todonotes}

\newtheorem{theorem}{Theorem}
\newtheorem{lemma}[theorem]{Lemma}
\newtheorem{corollary}[theorem]{Corollary}

\newtheorem{example}[theorem]{Example}
\newtheorem{remark}[theorem]{Remark}
\newtheorem{definition}[theorem]{Definition}

\newcommand{\Point}{\operatorname{sum}}
\newcommand{\sumset}{\Sigma}

\renewcommand{\nmid}{\nmid}   % 不整除符号
\begin{document}

\author{Haojie Chen\thanks{Haojie Chen is with the School of Mathematics, Sun Yat-sen University, Guangzhou 510275, China (e-mail: chenhj69@mail2.sysu.edu.cn).}, Chuangqiang Hu\(\dagger\) \thanks{Chuangqiang Hu is with the School of Mathematics, Sun Yat-sen University, Guangzhou 510275, China (e-mail: huchg5@mail.sysu.edu.cn).}, Junjie Huang\thanks{Junjie Huang is with the School of Mathematics, Sun Yat-sen University, Guangzhou 510275, China (e-mail: huangjj76@mail2.sysu.edu.cn).}, 
Chang-An Zhao\thanks{Chang-An Zhao is with the School of Mathematics, Sun Yat-sen University, Guangzhou 510275, China, and also with the Guangdong Key Laboratory of Information Security Technology, Guangzhou 510006, China (e-mail: zhaochan3@mail.sysu.edu.cn).}
\thanks{\(\dagger\) Corresponding author. } }

\title{On the Maximal Length of MDS Elliptic Codes}

\maketitle

\begin{abstract}
The determination of the maximal length of maximum distance separable (MDS) codes arising from elliptic curves is a central problem in coding theory. For an elliptic curve $E$ over $\mathbb{F}_q$, let $\operatorname{MEC}(k,q)$ denote the maximal length of a $q$-ary MDS elliptic code of dimension $k$. 
It was recently shown that $\operatorname{MEC}(k,q)\le\frac{q+1}{2}+\sqrt{q}$ for $q\ge289$ and $3\le k\le(q+1-2\sqrt{q})/10$, with equality for odd $k$ when $q$ is an odd square. This paper investigates the remaining open cases, namely even dimension $k$, non-square $q$ and fields of characteristic $2$, and provides a complete resolution of the tightness question for the two natural parity regimes of $q+1+\lfloor 2\sqrt{q}\rfloor$. We prove that if the support of $G$ (used to define the code) consists of $\mathbb{F}_q$-rational points, the bound decreases to $\frac{q+1}{2}+\sqrt{q}-1$ for even $k$. Without this restriction, we construct MDS codes attaining $\frac{q+1}{2}+\sqrt{q}$ for even $k$. More generally, we establish $\operatorname{MEC}(k,q)=\frac{q+1+\lfloor2\sqrt{q}\rfloor}{2}$ when $q+1+\lfloor2\sqrt{q}\rfloor$ is even, and $\operatorname{MEC}(k,q)=\frac{q+\lfloor2\sqrt{q}\rfloor}{2}$ when it is odd.


\end{abstract}

\section{Introduction}
Maximum distance separable (MDS) codes form an optimal class of error-correcting codes for given parameters and occupy a central place in coding theory. Throughout, a linear code is denoted by $[n,k,d]$, where $n$ is the length, $k$ the dimension, and $d$ the minimum distance. By the Singleton bound, any such code satisfies $d \le n-k+1$, and equality defines an MDS code. Such codes are extremely rare and obey stringent parameter constraints. Despite their scarcity, MDS codes are of fundamental importance, with important applications in reliable communication, cryptography, quantum information processing, and combinatorial design.

Let $M(k,q)$ denote the maximal length of a non‑trivial $q$-ary MDS code of dimension $k$. The long‑standing Main Conjecture on MDS codes (see \cite{MR465509}) states that for $q>k$, 
\[
M(k,q)=\begin{cases}
q+2, & q\text{ even and }k=3\text{ or }k=q-1,\\
q+1, & \text{otherwise}.
\end{cases}
\]

Considerable progress has been made toward this conjecture. It was shown in \cite{MR2911882} that the conjecture holds for prime fields. For $q$ an odd square, \cite{MR3846204,MR4073887} verified the conjecture for $k\le \sqrt{q} - \frac{\sqrt{q}}{p} + 2$. The case $11\le q\le 19$ was proved in \cite{MR1477225}. For recent progress, we refer to \cite{MR1398077,MR985681,MR1198385} and the references therein. The case of elliptic and hyperelliptic curves for sufficiently large $q$ was also addressed in \cite{MR1178199,MR1412183,MR1289576,MR1454971}.

Denote by $\mathbb{F}_q$ the finite field with $q$ elements.
Let $X$ be an absolutely irreducible smooth projective curve of genus $g$ over $\mathbb{F}_q$, let $\mathcal{D}=\{P_1,\dots,P_n\}$ be a set of $\mathbb{F}_q$-rational points, and let $G$ be an $\mathbb{F}_q$-rational divisor with $\operatorname{Supp}(G)\cap\mathcal{D}=\emptyset$ and $2g-2<\deg G<n$. Then the algebraic geometry code $C(X,\mathcal{D},G)$ (see Definition~\ref{def:definition_of_AG_codes}) is a linear $[n,k,d]$ code with $k=\deg G-g+1$. When $X=E$ is an elliptic curve, $C(E,\mathcal{D},G)$ is called an elliptic code. The main conjecture has been verified for elliptic curves when $q$ is sufficiently large (see, e.g., \cite{MR1178199,MR1412183}).

Several bounds on the maximal length of non-trivial MDS codes arising from elliptic curves are known.  
Define $\operatorname{MEC}(k,q)$ to be the maximal length of a non‑trivial $q$-ary MDS elliptic code of dimension $k$.  
Munuera \cite{MR1178199} established the bound  
\[
\operatorname{MEC}(k,q)\le \frac{q+1}{2}+\sqrt{q}+k.
\]
Subsequently, Li et al. \cite{7282884} obtained another bound
\[
\operatorname{MEC}(k,q)\le \Bigl(\frac23+\epsilon\Bigr)q
\]
for any $\epsilon>0$, provided that $q>4/\epsilon^2$ and $k>C_\epsilon\ln q$ (where $C_\epsilon$ is a constant depending on $\epsilon$). This bound improves upon Munuera's result for large $k$. Furthermore, Li, Wan, and Zhang (see \cite{7282884}, p.~2) conjectured that for every $\epsilon>0$ there exists some constant $C_\epsilon>0$ such that for all $q>C_\epsilon$,
\[
\operatorname{MEC}(k,q)\le \Bigl(\frac12+\epsilon\Bigr)q.
\]

In \cite{MR4564628} the conjecture for the range $3\le k\le \frac{q+1-2\sqrt{q}}{10}$ was proved, and a complete proof was subsequently provided in \cite{MR4756098}. The method of \cite{MR4564628} builds upon the ideas developed in \cite{7282884}, reducing the MDS condition for elliptic codes to a purely combinatorial problem (Lemma~\ref{lem:MDS_condition}) and relying on combinatorial results (see Lemmas~\ref{lem:tight:first_combination_result} and~\ref{lem:G_combination}). In particular, it was shown in \cite{MR4564628} that for $q\ge 289$ and $3\le k\le \frac{q+1-2\sqrt{q}}{10}$,
\begin{equation}\label{eq:bound}
    \operatorname{MEC}(k,q)\le \frac{q+1}{2}+\sqrt{q}.
\end{equation} 
Notably, this bound is tight \textbf{when $k$ is odd and $q$ is an odd square}.

{\color{blue}Recently, for an elliptic curve $E$ with an even number $N$ of $\mathbb{F}_q$-rational points, Wang, Liu, Luo, and Zhai (see \cite{wang2025nonreedsolomontypemdscodes}) constructed MDS elliptic codes of length $N/2$ when $k$ is odd, and of length $N/2-1$ when $k$ is even. In analogy with extended Reed–Solomon codes, they also extended the above MDS elliptic codes for even dimension $k$ to obtain linear codes of length $N/2$.}

However, several questions remain unresolved. In particular, the following problems have not been settled in the literature.
\begin{enumerate}
    \item \label{Q:1}It is not known whether the bound \eqref{eq:bound}
    % $\frac{q+1}{2}+\sqrt{q}$ 
    can be attained when $k$ is even and $q$ is an odd square.
    \item \label{Q:2}When $q$ is not a square, the exact value of $\operatorname{MEC}(k,q)$ remains unknown.
    \item \label{Q:3}When $\operatorname{char}(\mathbb{F}_q)=2$, the bound $\frac{q+1}{2}+\sqrt{q}$ in \eqref{eq:bound} is not an integer; consequently, the exact maximal length remains undetermined.
\end{enumerate}

This paper resolves the three open problems listed above completely. The resulting bounds to answer the three questions, combined with the earlier work of \cite{MR4564628}, are collected in Table~\ref{tab:MEC_max_length}.
\begin{table}[htbp]
  \centering
  \caption{Maximal length of non‑trivial MDS elliptic codes for $q \ge 289$ and $3 \le k \le \frac{q+1-2\sqrt{q}}{10}$}
  \label{tab:MEC_max_length}
  \renewcommand{\arraystretch}{1.6}
  \begin{tabular}{|c|c|c|c|c|}
    \hline
    \textbf{Maximal length} $n$ & \textbf{Condition on} $q$ & \textbf{Dimension} $k$ & \textbf{Restriction on} $G$ & \textbf{Citation} \\
    \hline
    $\frac{q+1+ 2\sqrt{q}}{2}$ & odd square $q$ & odd & none & \cite{MR4564628}, Equation \eqref{eq:bound} \\
    \hline
    $\frac{q+1+ \lfloor 2\sqrt{q}\rfloor}{2}$ & $q+1+ \lfloor 2\sqrt{q}\rfloor$ even & odd & none & Theorem~\ref{thm:q+1+...even} \\
    \hline
    $\frac{q+1+ \lfloor 2\sqrt{q}\rfloor}{2}-1$ & $q+1+ \lfloor 2\sqrt{q}\rfloor$ even & even & $\operatorname{Supp}(G) \subseteq E(\mathbb{F}_q)$ & Theorem~\ref{thm:G'and_even_k} \\
    \hline
    $\frac{q+1+ \lfloor 2\sqrt{q}\rfloor}{2}$ & $q+1+ \lfloor 2\sqrt{q}\rfloor$ even & even & none & Theorem~\ref{thm:q+1+...even} \\
    \hline
    $\frac{q+ \lfloor 2\sqrt{q}\rfloor}{2}$ & $q+1+ \lfloor 2\sqrt{q}\rfloor$ odd & odd & none & Theorem~\ref{thm:q+1+odd} \\
    \hline
    $\frac{q+ \lfloor 2\sqrt{q}\rfloor}{2}-1$ & $q+1+ \lfloor 2\sqrt{q}\rfloor$ odd & even & $\operatorname{Supp}(G) \subseteq E(\mathbb{F}_q)$ & Theorem~\ref{thm:G'and_oddq+1+} \\
    \hline
    $\frac{q+ \lfloor 2\sqrt{q}\rfloor}{2}$ & $q+1+ \lfloor 2\sqrt{q}\rfloor$ odd & even & none & Theorem~\ref{thm:q+1+odd} \\
    \hline
    $2^{a-1} + \sqrt{2^a} - 1$ & $q=2^a$ with even $a$ & even & $\operatorname{Supp}(G) \subseteq E(\mathbb{F}_q)$ & Corollary~\ref{cor:q=2^a,evenkandG'} \\
    \hline
    $2^{a-1} + \sqrt{2^a} $ & $q=2^a$ with even $a$ & even & none & Corollary~\ref{cor:boundforcharq=2} \\
    \hline
    $2^{a-1} + \sqrt{2^a} $ & $q=2^a$ with even $a$ & odd & none & Corollary~\ref{cor:boundforcharq=2} \\
    \hline
  \end{tabular}
\end{table}

In the following, we discuss each of the three problems in detail.

\subsection{For the First Problem}
In the existing constructions of MDS elliptic codes, the divisor $G$ is invariably taken to consist solely of $\mathbb{F}_q$-rational points. Accordingly, our initial approach was to seek an $[\frac{q+1+2\sqrt{q}}{2},k]$ MDS code of even dimension $k$ with $G$ likewise supported on $\mathbb{F}_q$-rational points. This attempt, however, proves unsuccessful. We shall establish the following result (see Corollary~\ref{cor:MEC-1}). Let $q \ge 289$ and assume that the support of $G$ is contained in the set of $\mathbb{F}_q$-rational points. Then for every even integer $k$ satisfying $3 \le k \le (q+1-2\sqrt{q})/10$, there exists no MDS elliptic code of length exactly $\frac{q+1}{2}+\sqrt{q}$. Equivalently, the bound $\frac{q+1}{2}+\sqrt{q}$ fails to be tight in this case. In Theorem~\ref{thm:G'and_even_k} we prove that the maximal length of a non‑trivial MDS elliptic code in this setting is precisely
\[
\frac{q+1}{2}+\sqrt{q}-1.
\]

Therefore, to construct an MDS elliptic code of length $\frac{q+1}{2}+\sqrt{q}$ and even dimension $k$, the support of $G$ must contain at least one place of degree strictly greater than $1$. We present such a construction in Theorem~\ref{thm:construction_of_degree_greater_1}. Consequently, the bound $\frac{q+1}{2}+\sqrt{q}$ is attainable for even $k$ when $q$ is an odd square, which answers the first question.

\subsection{For the Second Problem}
For the second problem, we split the discussion according to the parity of $q+1+ \lfloor 2\sqrt{q} \rfloor$.
As before, we first consider the case where the support of $G$ consists entirely of $\mathbb{F}_q$-rational points.
If $q+1+ \lfloor 2\sqrt{q} \rfloor$ is even, Theorem~\ref{thm:G'and_even_k} gives the tight upper bound
\[
\frac{q+1+ \lfloor 2\sqrt{q} \rfloor}{2} - 1 .
\]
If $q+1+ \lfloor 2\sqrt{q} \rfloor$ is odd, the analogous bound is
\[
\frac{q+ \lfloor 2\sqrt{q} \rfloor}{2} - 1 ,
\]
and its tightness is proved in Theorem~\ref{thm:G'and_oddq+1+}.

We now remove the restriction on $G$ and treat the two parity cases separately.

\noindent\textbf{Case 1: $q+1+ \lfloor 2\sqrt{q} \rfloor$ is even.}
With the aid of the construction in Theorem~\ref{thm:construction_of_degree_greater_1}, we obtain the following tight bound (see Theorem~\ref{thm:q+1+...even}):
\[
\operatorname{MEC}(k,q) = \frac{q+1+ \lfloor 2\sqrt{q} \rfloor}{2}.
\]
This covers, in particular, the previously open case where $q$ is an odd square and $k$ is even, thereby answering the first question posed at the outset.

\noindent\textbf{Case 2: $q+1+ \lfloor 2\sqrt{q} \rfloor$ is odd.}
In this situation, the relevant result is stated in Theorem~\ref{thm:q+1+odd}:
\[
\operatorname{MEC}(k,q) = \frac{q+ \lfloor 2\sqrt{q} \rfloor}{2}.
\]
This case is of particular importance for fields of characteristic $2$ (where $q = 2^{t}$ is square and the expression $q+1+ \lfloor 2\sqrt{q} \rfloor$ is naturally odd), providing the exact maximal length in a setting widely used in applications.

Thus, for every $q \ge 289$ and $3 \le k \le (q+1-2\sqrt{q})/10$, the exact value of $\operatorname{MEC}(k,q)$ is now completely determined, regardless of whether $q$ is a square or not.

\subsection{For the Third Problem}
Let $q = 2^{a} \ge 289$ and $3 \le k \le (q+1-2\sqrt{q})/10$.
 By applying the parity analysis developed for the second problem, we obtain the following definitive answer. If $2^{a}+1+\lfloor 2\sqrt{2^{a}}\rfloor$ is odd, then Corollary~\ref{cor:boundforcharq=2} gives
\[
\operatorname{MEC}(k,q) = \frac{2^{a} + \lfloor 2\sqrt{2^{a}} \rfloor}{2}.
\]
In particular, when $q$ is a square (i.e., $a$ is even), this simplifies to
\[
\operatorname{MEC}(k, q) = 2^{a-1} + \sqrt{2^a}.
\]

If $2^{a}+1+\lfloor 2\sqrt{2^{a}}\rfloor$ is even, then we have
\[
\operatorname{MEC}(k,q) = \frac{2^{a}+1 + \lfloor 2\sqrt{2^{a}} \rfloor}{2}.
\]
Thus, the exact maximal length for characteristic $2$ is completely determined, which resolves the third open question. These results are especially relevant for applications in cryptography and communication systems, where binary fields are predominantly employed.

This paper is organized as follows. Section~\ref{sec:preliminaries} introduces the necessary notation and basic properties of elliptic codes that will be used throughout. We treat the upper bound problem in two separate cases, depending on whether \( q+1 + 2\sqrt{q} \) is even or odd. Section~\ref{sec:A_tight_upper_bound_for_even_k} deals with the even case, and Section~\ref{sec:A_tight_upper_bound_for_the_case_of_characteristic_2} covers the odd case. In Section~\ref{sec:A_tight_upper_bound_for_even_k}, we first show that for certain even integers \(k\), the codes \([\frac{q+1+2\sqrt{q}}{2}, k]\) cannot be MDS under the restriction that the support of $G$ consists entirely of $\mathbb{F}_q$-rational points. We then present a new construction of MDS elliptic codes and derive a tight upper bound for \(\operatorname{MEC}(k,q)\) when \(q+1+ \lfloor 2\sqrt{q} \rfloor\) is even. This section also resolves the question of whether the bound \(\frac{q+1+2\sqrt{q}}{2}\) is attainable for even \(k\) in the specified range. In Section~\ref{sec:A_tight_upper_bound_for_the_case_of_characteristic_2},  we establish a tight upper bound for \(\operatorname{MEC}(k,q)\) when \(q+1+ \lfloor 2\sqrt{q} \rfloor\) is odd. Given the importance of characteristic‑\(2\) fields in communication and cryptography, we also specialize the bounds to this characteristic.

\section{Preliminaries}\label{sec:preliminaries}

In this section, we briefly recall the basic definitions and concepts pertaining to algebraic geometry codes and elliptic curves, as well as certain combinatorial results over finite abelian groups that will be required in the sequel.

\subsection{Algebraic Geometry Codes and Elliptic Curves}

Let $\mathbb{F}_q$ denote the finite field with $q$ elements and fix an algebraic closure $\overline{\mathbb{F}}_q$ of $\mathbb{F}_q$. Let $X$ be a smooth, projective, geometrically irreducible curve of genus $g$ defined over $\mathbb{F}_q$, and let $X(\overline{\mathbb{F}}_q)$ be the set of $\overline{\mathbb{F}}_q$-rational points on $X$. 
The divisor group $\operatorname{Div}(X)$ is the free abelian group generated by the points of $X(\overline{\mathbb{F}}_q)$. 
Thus, a divisor $G$ is a formal sum $G = \sum n_Q [Q]$, where $Q$ runs over $X(\overline{\mathbb{F}}_q)$ and the coefficients $n_Q$ are integers with $n_Q = 0$ for all but finitely many $Q$. The degree of $G$ is defined by $\deg(G) = \sum n_Q$. If $G$ is invariant under the action of $\operatorname{Gal}(\overline{\mathbb{F}}_q / \mathbb{F}_q)$, it is called an $\mathbb{F}_q$-rational divisor. The support of a divisor $G = \sum n_Q [Q]$ is $\operatorname{Supp}(G) = \{ Q : n_Q \neq 0 \}$.

Let $\overline{\mathbb{F}}_q(X)$ (resp. $\mathbb{F}_q(X)$) denote the function field of $X$ over $\overline{\mathbb{F}}_q$ (resp. over $\mathbb{F}_q$). For any nonzero function $f \in \overline{\mathbb{F}}_q(X)^{\times}$, its principal divisor is $\operatorname{div}(f) = \sum v_Q(f)[Q]$, where $v_Q(f)$ denotes the valuation of $f$ at the point $Q$. Given an $\mathbb{F}_q$-rational divisor $G$ on $X$, the associated Riemann-Roch space over $\mathbb{F}_q$ is defined as
\[
L(G) := \{ f \in \mathbb{F}_q(X)^{\times} : \operatorname{div}(f) + G \ge 0 \} \cup \{0\}.
\]
A place $P$ of the function field $\mathbb{F}_q(X)$ is the maximal ideal of a discrete valuation ring $\mathcal{O}_P \subseteq \mathbb{F}_q(X)$. The degree of a place $P$, denoted $\deg(P)$, is the degree of its residue field $[\mathcal{O}_P/P : \mathbb{F}_q]$.

There is a natural bijection between the places of $\mathbb{F}_q(X)$ and the Galois orbits of points in $X(\overline{\mathbb{F}}_q)$, under which a place of degree $d$ corresponds to a set of $d$ distinct conjugate points over $\overline{\mathbb{F}}_q$. In particular, an $\mathbb{F}_q$-rational point on $X$ is precisely a place of degree one. Throughout this paper we freely identify a place with the corresponding point (or set of conjugate points).

Algebraic geometry codes are constructed as follows.

\begin{definition}\label{def:definition_of_AG_codes}
Let $\mathcal{D} = \{P_1, \ldots, P_n\}$ be a set of $\mathbb{F}_q$-rational points on $X$, and let $G$ be an $\mathbb{F}_q$-rational divisor on $X$ such that $\operatorname{Supp}(G) \cap \mathcal{D} = \emptyset$. The algebraic geometry code (AG code) on $X$ associated with $\mathcal{D}$ and $G$ is defined by
\[
C(X, \mathcal{D}, G) = \left\{ \bigl( f(P_1), \ldots, f(P_n) \bigr) : f \in L(G) \right\} \subseteq \mathbb{F}_q^n.
\]
\end{definition}

By definition, the code $C(X, \mathcal{D}, G)$ has length $n$. Under the condition $2g - 2 < \deg(G) < n$, the Riemann-Roch theorem implies that its dimension is given by $k = \deg(G) + 1 - g$. Let $d$ denote the minimum distance. Since every function in $L(G)$ has at most $\deg(G)$ zeros, it follows that
\[
n - \deg(G) \le d.
\]

Now consider the case where $X = E$ is an elliptic curve over $\mathbb{F}_q$. Since the genus of $E $ is $g = 1$, we obtain the bounds
\[
n - k \le d \le n - k + 1
\]
under the assumption that $0 < k = \deg(G) < n$.

It is well-known (see \cite[Ch. III, Proposition 3.4]{MR2514094}) that the points of $E(\overline{\mathbb{F}}_q)$ forms an abelian group with the identity $\mathcal{O}$. As usual, the symbol $+$ denotes the group addition and $mP$ denotes the scalar multiplication of an integer $m$ by a point $P$. 
For any points $Q_1, \ldots, Q_d \in E(\overline{\mathbb{F}}_q)$, the relation $Q_1 + \cdots + Q_d = \mathcal{O}$ holds if and only if
\[
\sum_{i=1}^d [Q_i] - d[\mathcal{O}] = \operatorname{div}(f)
\]
for some $f \in \overline{\mathbb{F}}_q(E)^{\times}$.  We define
\[
\Point: \operatorname{Div}(X)\to E(\overline{\mathbb{F}}_q) \qquad  \Point(G) = \sum n_Q Q
\]
for any divisor $G = \sum n_Q [Q]$.
In this notation, we have $
\Point([Q]) = Q$.

Let $\mathcal{D} = \{P_1, \ldots, P_n\}$ be a set of $\mathbb{F}_q$-rational points on $E$. The $k$-sumset of $\mathcal{D}$ is defined as

\begin{align}\label{def:k-sumset}
    \sumset_k(\mathcal{D}) := \left\{ Q \in E(\mathbb{F}_q) \; \middle| \; Q = \sum_{P \in \mathcal{A}} P \text{ for some } \mathcal{A} \subseteq \mathcal{D} \text{ with } |\mathcal{A}| = k \right\},
\end{align}
where the addition is taken in the finite abelian group $E(\mathbb{F}_q)$. Using these notations, we are able to introduce an important criterion for an elliptic code to be MDS.
\begin{lemma}[{\cite[Proposition 4.1]{7282884}}]\label{lem:MDS_condition}
Let $ k  \in \{ 1, \ldots, n-1\} $ be the degree of $G$.  Then $C(E, \mathcal{D}, G)$ is not an MDS code if and only if $\Point(G) \in \sumset_k(\mathcal{D})$.
\end{lemma}
For a fixed elliptic curve $E$ over $\mathbb{F}_q$, let $\operatorname{MEC}(k, q, E)$ denote the maximal length of a non‑trivial $q$-ary MDS elliptic code of dimension $k$ constructed on the elliptic curve $E$. 
Define $\operatorname{MEC}(k,q)$ to be the maximal length of a non‑trivial $q$-ary MDS elliptic code of dimension $k$. We define $\operatorname{MEC}^*(k,q)$ in the same fashion under the restriction that the support of $G$ consists entirely of $\mathbb{F}_q$-rational points. 
Obviously, we have 
\[
\operatorname{MEC}^*(k,q) \le \operatorname{MEC}(k,q).
\]

\subsection{Some Results on the Rational Points of Elliptic Curves}
For positive integer $n$, let $E[n]$ denote the $n$-torsion subgroup of $E(\overline{\mathbb{F}}_q)$, consisting of points $P$ satisfying $nP = \mathcal{O}$. The following lemma provides a bound on the size of $E[n]$.

\begin{lemma}[{\cite[Chapter III, \S 6]{MR2514094}}]\label{lem:bound_on_torsion}
Let $E$ be an elliptic curve over $\mathbb{F}_q$. Then $|E[n]| \le n^2$.
\end{lemma}

Throughout this paper, we denote by $ N $ the cardinality of $ E(\mathbb{F}_q) $. 
The next lemma gives the classical Hasse-Weil bound on the cardinality $N$ of an elliptic curve.

\begin{lemma}[{\cite[Theorem 5.2.3]{MR2464941}}]\label{lem:the_bound_of_rational_number}
 Let $E$ be an elliptic curve over $\mathbb{F}_q$. Then $\bigl| N - q - 1 \bigr| \le 2\sqrt{q}$.
\end{lemma}
An elliptic curve is called maximal if its number of $\mathbb{F}_q$-rational points attains the upper bound $q + 1 + 2\sqrt{q}$.
The following lemmas provide explicit examples of maximal elliptic curves.
\begin{lemma}[\cite{MR3900980}]\label{lem:maximalcurveforp=2mod3}

\begin{itemize}
    \item[(i)] Let $p \neq 2$ be an odd prime with $p \equiv 2 \pmod{3}$. Then for any even $a$, there exists a maximal elliptic curve over $\mathbb{F}_{p^a}$ defined by $y^2 = x^3 + \theta^3$ for some $\theta \in \mathbb{F}_{p^a}^*$.
    \item[(ii)] Set $p = 3$. Then for any even $a$, there exists a maximal elliptic curve over $\mathbb{F}_{p^a}$ defined by $y^2 = x^3 + x$.
\end{itemize}
\end{lemma}

An elliptic curve $E/\mathbb{F}_q$ is called supersingular if its number of $\mathbb{F}_q$-rational points equals $1 + q + t$ for an integer $t$ divisible by the characteristic of $\mathbb{F}_q$. Two elliptic curves over $\mathbb{F}_q$ are said to be isogenous if they have the same number of $\mathbb{F}_q$-rational points.

The following lemma, due to Waterhouse, characterizes the possible orders of elliptic curves over $\mathbb{F}_q$.

\begin{lemma}[{\cite[Theorem 4.1]{MR265369}}]\label{lem:the_class_of_elliptic_curve}
 The isogeny classes of elliptic curves over $\mathbb{F}_q$ with $q = p^n$ are in one-to-one correspondence with the integers $t$ satisfying $|t| \le 2\sqrt{q}$ and one of the following conditions:
\begin{itemize}
    \item[(i)] $\gcd(t, p) = 1$;
    \item[(ii)] If $n$ is even: $t = \pm 2\sqrt{q}$;
    \item[(iii)] If $n$ is even and $p \not\equiv 1 \pmod{3}$: $t = \pm \sqrt{q}$;
    \item[(iv)] If $n$ is odd and $p = 2$ or $3$: $t = \pm p^{\frac{n+1}{2}}$;
    \item[(v)] If either (v1) $n$ is odd, or (v2) $n$ is even and $p \not\equiv 1 \pmod{4}$: $t = 0$.
\end{itemize}
The first case corresponds to ordinary (non-supersingular) curves, while the remaining cases are supersingular.
\end{lemma}

All possible group structures of $E(\mathbb{F}_q)$ were also determined in \cite{MR890272}.

\begin{lemma}[{\cite[Theorem 3]{MR890272}}]\label{lem:the_structure_of_group}
 Let $E$ be an elliptic curve over a finite field $\mathbb{F}_q$ with $q = p^n$. Assume that the number of rational points admits the prime factorization $N = \prod_l l^{h_l}$ for some integers $ h_l \geq 1 $.
 Then the possible group structure on $E(\mathbb{F}_q)$ is of the form:
\[
\mathbb{Z} / p^{h_p} \mathbb{Z} \times \prod_{l \neq p} \left( \mathbb{Z} / l^{a_l} \mathbb{Z} \times \mathbb{Z} / l^{h_l - a_l} \mathbb{Z} \right),
\]
where each $ a_l $ satisfies the following conditions:
\begin{itemize}
    \item[(a)] In case (ii) of Lemma~\ref{lem:the_class_of_elliptic_curve}: each $a_l = \frac{h_l}{2}$;
    \item[(b)] In cases (i), (iii), (iv), and (v) of Lemma~\ref{lem:the_class_of_elliptic_curve}: $a_l$ is an arbitrary integer satisfying
    \[
    0 \le a_l \le \min \left\{ v_l(q-1), \left\lfloor \frac{h_l}{2} \right\rfloor \right\},
    \]
    where $v_l(q-1)$ denotes the exponent of the prime $l$ in $q-1$ (i.e., the largest integer such that $l^{v_l(q-1)} \mid (q-1)$).
\end{itemize}
\end{lemma}

Lemmas~\ref{lem:the_class_of_elliptic_curve} and~\ref{lem:the_structure_of_group} together show that elliptic curves occur with many different orders $N$ and abelian group structures.

\subsection{Some Combinatorial Results on Finite Abelian Groups}

The group $E(\mathbb{F}_q)$ of an elliptic curve is a finite abelian group. We now recall several combinatorial results on finite abelian groups that will be used throughout. Let $\mathbf{G}$ be a finite abelian group and let $k$ be a positive integer. For any non‑empty subset $S \subseteq \mathbf{G}$,
the $k$-sumset of $S$ is defined as
\[
\sumset_k(S) := \left\{ g \in \mathbf{G} \; \middle| \; g = \sum_{a \in \mathcal{A}} a \text{ for some } \mathcal{A} \subseteq S \text{ with } |\mathcal{A}| = k \right\}.
\]
We also write $\mathbf{G}[2] := \{ g \in \mathbf{G} \mid 2g = 0 \}$ and $2\mathbf{G} := \{ 2g \in \mathbf{G} \mid g \in \mathbf{G} \}$. The following lemma is fundamental for our purposes.

\begin{lemma}[\cite{MR4564628}]\label{lem:tight:first_combination_result}
Let $S$ be a subset of $\mathbf{G}$ such that
\[
|S| > \max \left\{ \frac{2}{5}|\mathbf{G}|,\; 12|\mathbf{G}[2]| + 54 \right\}.
\]
Then either $\sumset_3(S) = \mathbf{G}$, or $S$ is contained in a coset of a subgroup of $\mathbf{G}$ of index $2$.
\end{lemma}

Lemma \ref{lem:tight:first_combination_result} provides a criterion for when a $3$-sumset equals the entire group.

\begin{lemma}[\cite{MR4564628}]\label{lem:G_combination}
 Let $\mathbf{G}$ be a finite abelian group with $|\mathbf{G}| \ge 30|\mathbf{G}[2]| + 135$, let $S$ be a subset of $\mathbf{G}$ satisfying $|S| > \frac{|\mathbf{G}|}{2}$, and let $k$ be a positive integer such that $3 \le k \le \frac{|\mathbf{G}|}{10}$. Then
\[
\sumset_k(S) =\mathbf{G}.
\]
\end{lemma}

Using the above lemma, the authors of \cite{MR4564628} established the following tight upper bound for the length of MDS codes of odd dimension and odd square $q$.

\begin{lemma}[\cite{MR4564628}]\label{lem:a_tight_upper_bound_before}
\begin{itemize}
    \item[(i)] Let $C(E, \mathcal{D}, G)$ be an $[n, k, d]$ MDS elliptic code. If $q \ge 289$ and $3 \le k \le \frac{N}{10}$, then $n \le \frac{N}{2}$.
    \item[(ii)] If $q \ge 289$ and $3 \le k \le \frac{q+1-2\sqrt{q}}{10}$, then
    \begin{equation} \label{eq:MEC}
        \operatorname{MEC}(k, q) \le \frac{q+1}{2} + \sqrt{q}.
    \end{equation}
    Moreover, the equality in \eqref{eq:MEC} holds when  $q$ is an odd square and $k$ is odd.
\end{itemize}
%  if $q$ is an odd square and $k$ is odd., then
% \[
% \operatorname{MEC}(k, q) = \frac{q+1}{2} + \sqrt{q}.
% \]
\end{lemma}

The following elementary lemma guarantees the existence of a subgroup of index $2$ in any finite abelian group of even order.

\begin{lemma}\label{lem:the_existence_of_index_two_subgroup}
Let $\mathbf{G}$ be a finite abelian group of even order. Then there exists a subgroup $H \le \mathbf{G}$ of index $2$.
\end{lemma}
\begin{proof}
This follows directly from the structure theorem of finite abelian groups.
% Let $\ord(\mathrm{G})$ denote the order of $\mathrm{G}$.
% Write $|\mathrm{G}| = 2^k n$, where $k \ge 1$ and $n$ is odd. We proceed by induction on $k$.

% For $k = 1$, the existence of an index-$2$ subgroup follows immediately from the structure theorem for finite abelian groups.

% Now suppose that the claim holds for all integers $m < k$. By Cauchy's theorem (or Sylow's theorems), $\mathrm{G}$ contains an element of order $2$; let $a$ be such an element. Then
% \[
% |\mathrm{G} / \langle a \rangle| = \frac{|\mathrm{G}|}{2} = 2^{k-1} n.
% \]
% Applying the induction hypothesis to $\mathrm{G} / \langle a \rangle$, we obtain a subgroup $H_1 \le \mathrm{G} / \langle a \rangle$ with $[\mathrm{G} / \langle a \rangle : H_1] = 2$.

% Let $\varphi: \mathrm{G} \to \mathrm{G} / \langle a \rangle$ be the canonical projection. By the correspondence theorem, there exists a unique subgroup $H \le \mathrm{G}$ containing $\langle a \rangle$ such that $H / \langle a \rangle = H_1$; equivalently, $H = \varphi^{-1}(H_1)$. This subgroup satisfies $[\mathrm{G} : H] = 2$, completing the proof.
\end{proof}

We shall need the following elementary fact concerning subgroups of order $2$ and subgroups of index $2$.
\begin{lemma}\label{lemma:order_two}
    Let $\mathbf{G}$ be a finite abelian group and let $\mathbf{G}[2] = \{g \in \mathbf{G} \mid 2g = 0\}$ be its $2$-torsion subgroup. Then the number of subgroups $H \le \mathbf{G}$ of index $2$ equals $|\mathbf{G}[2]| - 1$.
\end{lemma}

\begin{proof}
    Consider the homomorphism $\varphi: \mathbf{G} \to \mathbf{G}$ given by $\varphi(g) = 2g$. Its image is $2\mathbf{G} = \{2g \mid g \in \mathbf{G}\}$ and its kernel is $\mathbf{G}[2]$. Thus we have an exact sequence
    \[
    0 \longrightarrow \mathbf{G}[2] \longrightarrow \mathbf{G} \stackrel{\varphi}{\longrightarrow} \mathbf{G} \longrightarrow \mathbf{G}/2\mathbf{G} \longrightarrow 0,
    \]
    where the last map is the canonical projection. Since $\mathbf{G}$ is finite, we get
    \[
    |\mathbf{G}/2\mathbf{G}| = |\mathbf{G}[2]|.
    \]
    Moreover, $\mathbf{G}/2\mathbf{G}$ is an elementary abelian $2$-group (since $2(g+2\mathbf{G})=2g+2\mathbf{G}=0$ in the quotient), and the same holds for $\mathbf{G}[2]$. Therefore both are $\mathbb{F}_2$-vector spaces, and they have the same dimension, say $r$.

    Every subgroup $H \le \mathbf{G}$ of index $2$ is the kernel of a non‑trivial homomorphism $\chi: \mathbf{G} \to \{\pm 1\} \cong \mathbb{F}_2$ (viewed additively). Conversely, the kernel of any such non‑trivial character is a subgroup of index $2$. Distinct characters give distinct kernels because a non‑trivial character is uniquely determined by its kernel (the target has only two elements). Hence the number of index‑$2$ subgroups equals the number of non‑trivial elements of the dual group $\widehat{\mathbf{G}} = \operatorname{Hom}(\mathbf{G}, \mathbb{F}_2)$.

    For any finite abelian group, every homomorphism $\mathbf{G} \to \mathbb{F}_2$ factors through $\mathbf{G}/2\mathbf{G}$ because $2\mathbf{G}$ is contained in the kernel (as $\mathbb{F}_2$ is of exponent $2$). Thus $\widehat{\mathbf{G}} \cong \operatorname{Hom}(\mathbf{G}/2\mathbf{G}, \mathbb{F}_2)$. Since $\mathbf{G}/2\mathbf{G}$ is an $\mathbb{F}_2$-vector space of dimension $r$, its dual space is isomorphic to $\mathbf{G}/2\mathbf{G}$ itself, and therefore $|\widehat{\mathbf{G}}| = |\mathbf{G}/2\mathbf{G}| = 2^r = |\mathbf{G}[2]|$.

    Consequently, the number of non‑trivial characters, and hence the number of subgroups of index $2$, is $|\mathbf{G}[2]| - 1$.
\end{proof}

\section{A Tight Upper Bound When $q+1+ \lfloor 2\sqrt{q} \rfloor$ is Even}\label{sec:A_tight_upper_bound_for_even_k}
In this section, we determine the maximal length of non‑trivial MDS elliptic codes when $q+1+ \lfloor 2\sqrt{q} \rfloor$ is even, first under the assumption that the divisor $G$ is supported on $\mathbb{F}_q$-rational points, and then in the more general setting where this restriction is lifted.

We begin by deriving a necessary condition for MDS elliptic codes (Theorem~\ref{thm:main_theorem}). From this we deduce that when $\operatorname{Supp}(G)\subseteq E(\mathbb{F}_q)$ and $q$ is an odd square, the maximal length for even $k$ drops to $\frac{q+1+2\sqrt{q}}{2}-1$ (Theorem~\ref{thm:G'and_even_k} and Corollary~\ref{cor:MEC-1}). 

We then remove the restriction on $G$ and present an explicit construction that uses a place of degree $3$ in the support of $G$. This construction yields MDS codes of length $\frac{q+1+2\sqrt{q}}{2}$ and even dimension $k$ (Theorem~\ref{thm:construction_of_degree_greater_1}), showing that the original bound \eqref{eq:MEC} is again achievable. Finally, we combine the restricted and unrestricted cases to obtain a complete description of $\operatorname{MEC}(k,q)$ when $q+1+\lfloor 2\sqrt{q}\rfloor$ is even (Theorem~\ref{thm:q+1+...even}), which includes odd square $q$ and answers the open question \ref{Q:1} posed in the introduction.

We maintain the notation introduced above. Let $E$ be an elliptic curve defined over $\mathbb{F}_q$ and let $N=|E(\mathbb{F}_q)|$.

\subsection{Even-Dimensional MDS Codes of Length $N/2$}
% Let $C(E, \mathcal{D}, G)$ be a code on $E$ of length $n$ and dimension $k$. 
In this subsection, we derive a necessary condition for a code $C(E, \mathcal{D}, G)$ of even dimension $k$ to be MDS.

\begin{theorem}\label{thm:main_theorem}
    Assume that $C(E, \mathcal{D}, G)$ is an MDS code derived from the elliptic curve $ E $ over $\mathbb{F}_q$. Suppose that $ N = 2 |\mathcal{D}|$ and
    \[
    N \ge 30|E(\mathbb{F}_q)[2]| + 135.
    \]
    Moreover, the dimension $k = \deg G$ is even and satisfies
    $
    3 \le k \le \frac{N}{10}
    $. 
    Then the following statements hold.  
     \begin{enumerate}
     \item The set $\mathcal{D}$ is a coset of an index‑$2$ subgroup of $E(\mathbb{F}_q)$;
     \item The support of $G$ contains at least one place of degree greater than $1$.
     \end{enumerate}
\end{theorem}

\begin{proof}
We prove the contrapositive: if either $\mathcal{D}$ is not a coset of any index‑$2$ subgroup of $E(\mathbb{F}_q)$ or $\operatorname{Supp}(G)$ consists entirely of $\mathbb{F}_q$-rational points, then $C(E, \mathcal{D}, G)$ is not MDS.

\begin{enumerate}
    \item \label{item:1D}\textbf{$\mathcal{D}$ is not a coset of any index‑$2$ subgroup of $E(\mathbb{F}_q)$.} 
 Let $H_1, \dots, H_s$ be all index‑$2$ subgroups of $E(\mathbb{F}_q)$; by Lemma~\ref{lemma:order_two} we have $s = |E(\mathbb{F}_q)[2]| - 1$. By the assumption on $\mathcal{D}$, for each $i$ we have $\mathcal{D} \cap H_i \neq \emptyset$ and $\mathcal{D} \cap H_i^c \neq \emptyset$. Consequently, we may choose distinct elements
    \[
        a_i \in H_i \cap \mathcal{D}, \qquad b_i \in H_i^c \cap \mathcal{D}.
    \]
    Thus we can select at least $2s$ distinct elements from $\mathcal{D}$ with the above property.

    Choose a subset $A \subseteq \mathcal{D}$ that contains $\{a_i, b_i\}_{i=1}^s$ and satisfies
    \[
    |A| > \max\left\{ \frac{2}{5}N,\; 12|E(\mathbb{F}_q)[2]| + 54 \right\} = \frac{2}{5}N.
    \]
    In particular, we may take $|A| = \lfloor \frac{2}{5}N\rfloor + 1$. By Lemma~\ref{lem:tight:first_combination_result}, this guarantees $\sumset_3(A) = E(\mathbb{F}_q)$.

    From the hypotheses one checks that $k - 3 \le |\mathcal{D}| - |A|$. Indeed,
    \[
    |\mathcal{D}| - |A| \ge \frac{N}{2} - \frac{2}{5}N - 1 = \frac{N}{10} - 1,
    \]
    while $k \le N/10$, so $k - 3 \le N/10 - 3 < N/10 - 1$. Hence we can pick $k-3$ distinct points $p_1, \dots, p_{k-3}$ from $\mathcal{D} \setminus A$. Denote their sum by $h = p_1 + \cdots + p_{k-3} \in E(\mathbb{F}_q)$.

    For any $g \in E(\mathbb{F}_q)$, since $\sumset_3(A) = E(\mathbb{F}_q)$, there exist $a_1, a_2, a_3 \in A$ with
    \[
    a_1 + a_2 + a_3 = g - h.
    \]
    Consequently,
    \[
    p_1 + \cdots + p_{k-3} + a_1 + a_2 + a_3 = g,
    \]
    which shows $g \in \sumset_k(\mathcal{D})$. As $g$ was arbitrary, $\sumset_k(\mathcal{D}) = E(\mathbb{F}_q)$, and therefore $\Point(G) \in \sumset_k(\mathcal{D})$. By Lemma~\ref{lem:MDS_condition}, $C(E, \mathcal{D}, G)$ is not MDS.

    \item  \textbf{The support of $G$ consists entirely of $\mathbb{F}_q$-rational points.}
    We now prove that if the support of $G$ consists entirely of $\mathbb{F}_q$-rational points, then for every even integer $k$ with
\[
3 \le k \le \frac{N}{10},
\]

the code $C(E, \mathcal{D}, G)$ of length $|\mathcal{D}| = N/2$ cannot be MDS. By Lemma~\ref{lem:MDS_condition}, it suffices to show that
\[
\Point(G) \in \sumset_k(\mathcal{D}).
\]
If $\mathcal{D}$ is not a coset of any index-2 subgroup, then by Case~\ref{item:1D} the code is not MDS. Hence it remains to consider the situation where $\mathcal{D}$ is such a coset. We will show that even then the code cannot be MDS.
  % That is $\mathcal{D} = H$ or $\mathcal{D} = u + H$ for 
    % subgroup $H \le E(\mathbb{F}_q)$ of index two and
    \begin{itemize}
        \item \textbf{Subcase $\mathcal{D} = H$ is a subgroup of $E(\mathbb{F}_q)$ of index two.} Since $H$ is a subgroup of $E(\mathbb{F}_q)$, it is straightforward to check that
        \begin{equation}\label{eq:sumH<H}
            \sumset_k(H) = \sumset_k(\mathcal{D}) \subseteq H.
        \end{equation}
        The set $E(\mathbb{F}_q) \setminus E(\mathbb{F}_q)[2]$ can be partitioned into disjoint two‑element sets of the form $\{a, -a\}$ with $a \neq -a$. By Lemma~\ref{lem:bound_on_torsion}, $|E(\mathbb{F}_q)[2]| \le 4$, and consequently $H$ contains at least
        \[
        \frac{|H| - 4}{2}
        \]
        such pairs. For any $h \in H$, choose $\frac{k-2}{2}$ pairs $\{a_i, -a_i\}$ ($1 \le i \le \frac{k-2}{2}$) from these, avoiding $\mathcal{O}$ and $h$. Then
        \[
        h + \mathcal{O} + \sum_{i=1}^{(k-2)/2} \bigl(a_i + (-a_i)\bigr) = h,
        \]
        which shows $h \in \sumset_k(H)$. Since $h$ is arbitrary, $H \subseteq \sumset_k(H)$. Together with the reverse inclusion~\eqref{eq:sumH<H}, we obtain $H = \sumset_k(H)$.
        For even $k$, any $k$ elements from $u+H$ sum to an element of $H = \sumset_k(H)$, i.e., $\Point(G) \in \sumset_k(\mathcal{D})$.
        \item \textbf{Subcase $\mathcal{D} = u + H$ for some $u \in E(\mathbb{F}_q) \setminus H$.} Let $k u = h_2 \in H$ (note that $k$ is even, so $h_2 \in H$). Using the previous subcase,
        \[
        \sumset_k(u+H) = h_2 + \sumset_k(H) = h_2 + H = H.
        \]
        Hence $\Point(G) \in \sumset_k(u+H)$, and Lemma~\ref{lem:MDS_condition} implies that $C(E, \mathcal{D}, G)$ is not MDS.
    \end{itemize}
We have shown that under the assumption that $\operatorname{Supp}(G)$ consists of $\mathbb{F}_q$-rational points, the code $C(E, \mathcal{D}, G)$ can never be MDS.  
\end{enumerate}
% Consequently, a necessary condition for $C(E, \mathcal{D}, G)$ to be MDS is that $\mathcal{D}$ is a coset of an index‑$2$ subgroup of $E(\mathbb{F}_q)$ and the support of $G$ contains at least one place of degree strictly greater than $1$.  
% This completes the proof.
\end{proof}
The following result is an immediate corollary of Theorem~\ref{thm:main_theorem}.

\begin{corollary}\label{cor:evenE(Fq)}
    Under the hypotheses in Theorem \ref{thm:main_theorem}, if the support of $G$ consists entirely of $\mathbb{F}_q$-rational points, then the maximal length of a non‑trivial MDS elliptic code of dimension $k$ on $E$ satisfies the upper bound
    \begin{equation}\label{eq:MEC(k,q,E)leq N/2-1}
          \frac{N}{2} - 1.
    \end{equation}
\end{corollary}

\begin{proof}
    The proof follows directly from Lemma \ref{lem:a_tight_upper_bound_before} and Theorem \ref{thm:main_theorem}.
\end{proof}

\subsection{Case: $q+1+ \lfloor 2\sqrt{q} \rfloor$ Is Even and $\operatorname{Supp}(G)$ Consists of $\mathbb{F}_q$-Rational Points}

If $q$ is an odd square and $E$ is a maximal elliptic curve defined over $\mathbb{F}_q$, then Corollary~\ref{cor:evenE(Fq)} shows that the maximal length of a non‑trivial MDS elliptic code with even dimensional is bounded by $\frac{q+1+2\sqrt{q}}{2}-1$ (see \eqref{eq:MEC(k,q,E)leq N/2-1}) under the assumption that the support of $G$ consists entirely of $\mathbb{F}_q$-rational points. A natural question then arises: what is the exact maximal length under this restriction on $G$?
We now resolve the case where $q+1+ \lfloor 2\sqrt{q} \rfloor$ is even, which in particular covers the subcase of odd square $q$ treated above.
\begin{theorem}\label{thm:G'and_even_k}
Let $q \ge 289$ and assume that the support of $G$ consists entirely of $\mathbb{F}_q$-rational points. If $q+1+ \lfloor 2\sqrt{q} \rfloor$ is even, then for every even integer $k$ with $3 \le k \le (q + 1 - 2\sqrt{q})/10$, 
\begin{equation}\label{eq:G'MEC}
\operatorname{MEC}^*(k,q) \le \frac{q+1+ \lfloor 2\sqrt{q} \rfloor}{2} - 1.
\end{equation}
Moreover, the bound~\eqref{eq:G'MEC} is tight.
\end{theorem}

\begin{proof}
Let $E$ be an arbitrary elliptic curve over $\mathbb{F}_q$ and set $N = |E(\mathbb{F}_q)|$.  
From Lemma~\ref{lem:the_bound_of_rational_number} (the Hasse–Weil bound) we obtain
\[
q+1-2\sqrt{q} \le N \le q+1+2\sqrt{q}.
\]
The lower bound implies $k \le \frac{N}{10}$ (since $k \le (q+1-2\sqrt{q})/10$ by hypothesis) and, together with $q\ge 289$, yields
\[
N \ge 256 \ge 30|E(\mathbb{F}_q)[2]| + 135,
\]
where the final inequality follows from $|E(\mathbb{F}_q)[2]|\le 4$ (Lemma~\ref{lem:bound_on_torsion}).  

Now consider any code $C(E,\mathcal{D},G)$ on $E$ with $\deg G = k$ and $|\mathcal{D}| = N/2$.  
The conditions $N \ge 30|E(\mathbb{F}_q)[2]|+135$ and $3\le k\le N/10$ allow us to apply Theorem~\ref{thm:main_theorem}.  
Since $\operatorname{Supp}(G)$ is assumed to consist entirely of $\mathbb{F}_q$-rational points, the necessary condition stated in that theorem cannot be fulfilled (it would require a place of degree $>1$).  
Hence $C(E,\mathcal{D},G)$ is not MDS for any admissible $\mathcal{D}$ of size $N/2$.  
Therefore any MDS elliptic code on $E$ of dimension $k$ must have length at most $\frac{N}{2}-1$.  
Combining this with the upper bound $N \le q+1+\lfloor 2\sqrt{q}\rfloor$, we obtain precisely the bound \eqref{eq:G'MEC}.

To verify that \eqref{eq:G'MEC} is attainable, we need to construct an explicit MDS code of length $\frac{q+1+\lfloor 2\sqrt{q}\rfloor}{2}-1$. {\color{blue}The construction in \cite{wang2025nonreedsolomontypemdscodes} yields this result; for completeness, we provide an alternative construction here. } By Lemma~\ref{lem:the_class_of_elliptic_curve}, there exists an elliptic curve $E$ over $\mathbb{F}_q$ with
$N = |E(\mathbb{F}_q)| = q+1+\lfloor 2\sqrt{q}\rfloor$; in particular $N$ is even.  
Lemma~\ref{lem:the_existence_of_index_two_subgroup} guarantees a subgroup $H \subset E(\mathbb{F}_q)$ of index $2$.  
Choose an $\mathbb{F}_q$-rational point $u \notin H$ and define
\[
\mathcal{D} = (u + H) \setminus \{u\}, \qquad G = (k+1)[\mathcal{O}] - [u].
\]
By Lemma~\ref{lem:MDS_condition}, the code $C(E, \mathcal{D}, G)$ is MDS iff
\[
\Point\bigl((k+1)[\mathcal{O}] - [u]\bigr) \notin \sumset_k\bigl((u+H) \setminus \{u\}\bigr).
\]
Assume, for contradiction, that this condition fails.  
Then there exist distinct points $p_1, \dots, p_k \in \mathcal{D}$ such that
\[
\Point\bigl((k+1)[\mathcal{O}] - [u]\bigr) = p_1 + \cdots + p_k.
\]
By definition of $\Point(G)$, this is equivalent to
\[
\Point\bigl((k+1)[\mathcal{O}]\bigr) = p_1 + \cdots + p_k + u.
\]
The points $p_1,\dots,p_k,u$ lie in the coset $u+H$.  
Since $u+H$ is a coset of the index‑$2$ subgroup $H$ and $k+1$ is odd, the sum of an odd number of elements from such a coset can never be $\mathcal{O}$.  
This contradiction shows that the MDS condition is satisfied, and the constructed code achieves length $\frac{N}{2}-1 = \frac{q+1+\lfloor 2\sqrt{q}\rfloor}{2}-1$.  
Thus the bound \eqref{eq:G'MEC} is tight.
\end{proof}

In particular, $q+1+ \lfloor 2\sqrt{q} \rfloor$ is always even whenever $q$ is an odd square. Consequently, we obtain the following corollary.

\begin{corollary}\label{cor:MEC-1}
Let $q \ge 289$ be an odd square and assume that the support of the divisor $G$ consists entirely of $\mathbb{F}_q$-rational points. Then for every even integer $k$ with $3 \le k \le (q + 1 - 2\sqrt{q})/10$,
\begin{equation}\label{eq:3bound}
    \operatorname{MEC}^*(k,q) = \frac{q+1+2\sqrt{q}}{2} - 1.
\end{equation}

\end{corollary}

\begin{remark}\label{remark:suppG_contain_gt1}
For completeness, we note that if the restriction on the support of $G$ is removed, then for even $k$ we still have
\[
\operatorname{MEC}(k, q, E) \le \frac{N}{2},
\]
as shown in Lemma~\ref{lem:a_tight_upper_bound_before}. It will be proved later in Theorem~\ref{thm:construction_of_degree_greater_1} that this bound is tight when $k$ is even and $N$ is even.
\end{remark}

\subsection{A Construction of MDS Elliptic Codes}
In the preceding discussion, we imposed the restriction that the support of $G$ consists entirely of $\mathbb{F}_q$-rational points, and under this restriction we obtained the tight upper bound \eqref{eq:G'MEC} and \eqref{eq:3bound} for the maximal length of non‑trivial MDS codes. In this subsection, we present a construction of MDS codes showing that, when no restriction is imposed on $G$, MDS codes of length $\frac{q+1+2\sqrt{q}}{2}$ and even dimension $k$ do exist.

\begin{theorem}\label{thm:construction_of_degree_greater_1}
Let $E$ be an elliptic curve over $\mathbb{F}_q$ such that number of rational points $N = |E(\mathbb{F}_q)|$ is even. Assume that there exists a 
place of odd degree $l>2 $ with $ \Point([R])  =  \mathcal{O}$. 
% function $f \in \mathbb{F}_q(E)$ whose divisor is
% \[
% \operatorname{div}(f) = [R] - l[\mathcal{O}],
% \]
% where $l$ is an odd positive integer and $R$ is a place of degree $l$. 
Then  for any integer $k \ge l$, there exists an MDS code $C(E, \mathcal{D}, G)$ of length $N/2$ and dimension $k$.
\end{theorem}
\begin{proof}
Since $N$ is even, Lemma~\ref{lem:the_existence_of_index_two_subgroup} guarantees the existence of a subgroup $H \subset E(\mathbb{F}_q)$ of index $2$, and hence of order $N/2$. 
\begin{enumerate}
    \item \textbf{$k$ is even.} Choose an element $u \notin H$, and let $P_1, \dots, P_{k-l}$ be any $k-l$ elements from the coset $u+H$. We take $\mathcal{D} = H$ and 
\[
 \qquad G = [R] + [P_1] + \cdots + [P_{k-l}].
\]
Clearly, $C(E, \mathcal{D}, G)$ has length $|H| = N/2$ and dimension $k$. It remains to verify that this code is MDS.

Applying $
\Point([R]) = \mathcal{O}
$, we get
\begin{align*}
\Point(G) &= \Point({[R] + [P_1] + \cdots + [P_{k-l}]}) \\
% & = \Point([R]) + P_1 + \cdots + P_{k-l} \\
&= P_1 + \cdots + P_{k-l}.
\end{align*}
Since each $P_i$ lies in $u+H$ and $k-l$ is odd, their sum $\Point(G)$ belongs to $u+H$ as well (the sum of an odd number of elements of a coset of an index‑$2$ subgroup remains in that coset). Hence $\Point(G) \notin H = \mathcal{D}$, which gives $\Point(G) \notin \sumset_k(\mathcal{D})$. By Lemma~\ref{lem:MDS_condition}, the code $C(E, \mathcal{D}, G)$ is MDS.
    
    \item \textbf{$k$ is odd.} Choose an element $u \notin H$, and let $P_1, \dots, P_{k-l}$ be any $k-l$ elements taken from $H$. Define
\[
\mathcal{D} = u + H, \qquad G = [R] + [P_1] + \cdots + [P_{k-l}].
\]
A similar argument shows that the resulting code $C(E, \mathcal{D}, G)$ is MDS.
Applying $
\Point([R]) = \mathcal{O}
$ again, we get
\begin{align*}
\Point(G) &= \Point({[R] + [P_1] + \cdots + [P_{k-l}]}) \\
% & = \Point([R]) + P_1 + \cdots + P_{k-l} \\
&= P_1 + \cdots + P_{k-l}.
\end{align*}
Since each $P_i$ lies in $H$, their sum $\Point(G)$ belongs to $H$ as well. Hence $\Point(G) \notin \sumset_k(\mathcal{D})$. By Lemma~\ref{lem:MDS_condition}, the code $C(E, \mathcal{D}, G)$ is MDS, which completes the proof.
\end{enumerate}

\end{proof}

In order to apply Theorem~\ref{thm:construction_of_degree_greater_1}, we show that the desired place of degree three exists under the given hypotheses $q \ge 6$.

\begin{lemma}\label{lem:existenceofdegree3place}
Let $q \ge 6$ be a prime power, and let $E$ be an elliptic curve defined over $\mathbb{F}_q$. Then there exists a place $R$ of degree three such that
$
\Point([R]) = \mathcal{O}$.
\end{lemma}

\begin{proof}
Let $\pi$ denote the Frobenius endomorphism. Define the map $\operatorname{T}: E(\mathbb{F}_{q^3}) \to E(\mathbb{F}_q)$ by
\[ \operatorname{T}(P) = P + \pi(P) + \pi^2(P).\]
Since $\pi$ is a group homomorphism, $\operatorname{T}$ is a group homomorphism as well. The kernel of $\operatorname{T}$ is
\[
\operatorname{Ker}(\operatorname{T}) = \{ P \in E(\mathbb{F}_{q^3}) \mid P + \pi(P) + \pi^2(P) = \mathcal{O} \}.
\]
 The cardinality of $\operatorname{Ker}(\operatorname{T})$ satisfies
\[
|\operatorname{Ker}(\operatorname{T})|
= \frac{|E(\mathbb{F}_{q^3})|}{|\operatorname{Im}(\operatorname{T})|} \ge \frac{|E(\mathbb{F}_{q^3})|}{N} .
\]
Applying Lemma~\ref{lem:the_bound_of_rational_number}, it follows
\begin{align*}
|\operatorname{Ker}(\operatorname{T})| &\ge \frac{q^3 + 1 - 2\sqrt{q^3}}{q + 1 + 2\sqrt{q}} \\
% &= \left( \frac{(\sqrt{q} - 1)(q + \sqrt{q} + 1)}{\sqrt{q} + 1} \right)^{\!2} \\
% &\ge \bigl( (\sqrt{q} - 1)\sqrt{q} \bigr)^2 \\
&\ge 10.
\end{align*}
It is evident that $ \operatorname{T}(P) = 3P  $ for $P \in  E(\mathbb{F}_q)$.
Hence,
\[
\operatorname{Ker}(\operatorname{T}) \cap E(\mathbb{F}_q) \subseteq E[3].
\]
Lemma~\ref{lem:bound_on_torsion} yields
\[
|\operatorname{Ker}(\operatorname{T}) \cap E(\mathbb{F}_q)| \le 9.
\]
Since $|\operatorname{Ker}(\operatorname{T})| \ge 10$, there exists a point $P \in \operatorname{Ker}(\operatorname{T})$ with $P \notin E(\mathbb{F}_q)$.
The three conjugate points $P$, $\pi(P)$, $\pi^2(P)$ are pairwise distinct and together form a place $ R $ of degree $3$ on $E/\mathbb{F}_q$.
The place $R$ fulfills the required condition, completing the proof.
\end{proof}
For the explicit examples below, the degree‑$3$ place $R$ is found via a more efficient, alternative method.

\begin{example}\label{ex:p=2mod3}
    \begin{enumerate}
        \item Let $p$ be an odd prime with $p \equiv 2 \pmod{3}$, and let $q = p^a$ with $a$ even. By Lemma~\ref{lem:maximalcurveforp=2mod3}, there exists a maximal elliptic curve over $\mathbb{F}_q$ defined by 
        \[
        y^2 = x^3 + \theta^3
        \]
        for some $\theta \in \mathbb{F}_{q}^*$. From $p \equiv 2 \pmod{3}$ we obtain $3 \mid (q-1)$.
        Consider the maps $\phi(t)=t^3+\theta^3$ and $\Phi(t)=t^2$ from $\mathbb{F}_q$ to $\mathbb{F}_q$. Since $3 \mid (q-1)$, the image of $\phi$ has size $\frac{q-1}{3}+1$; since $2 \mid (q-1)$, the image of $\Phi$ has size $\frac{q-1}{2}+1$.
        
        Hence there exists an element $b \in \mathbb{F}_q$ such that the equation $b^2 = x^3 + \theta^3$ admits no solution in $\mathbb{F}_q$. This yields a place $R$ of degree $3$ with $\operatorname{div}(y-b) = [R] - 3[\mathcal{O}]$.
        
        \item Let $p = 3$ and let $q = p^a$ with $a$ even. By Lemma~\ref{lem:maximalcurveforp=2mod3}, there exists a maximal elliptic curve over $\mathbb{F}_q$ defined by 
        \[
        y^2 = x^3 + x.
        \]
        Define $\phi: \mathbb{F}_q \to \mathbb{F}_q$ by $\phi(t) = t^3 + t$. Since $\phi$ is an $\mathbb{F}_3$-linear map, its image has size $\frac{q}{3}$.
        
        Define $\Phi: \mathbb{F}_q \to \mathbb{F}_q$ by $\Phi(t) = t^2$. Since $2 \mid (q-1)$, the image of $\Phi$ has size $\frac{q-1}{2} + 1$.
        
        Consequently, there exists an element $b \in \mathbb{F}_q$ such that $b^2 = x^3 + x$ admits no solution in $\mathbb{F}_q$, which gives a place $R$ of degree $3$ with $\operatorname{div}(y - b) = [R] - 3[\mathcal{O}]$.
    \end{enumerate}
\end{example}

For the two families of elliptic curves in Example~\ref{ex:p=2mod3}, the search for the place $R$ is confined to a set of size roughly $q$. In contrast, the method of Lemma~\ref{lem:existenceofdegree3place} requires working in $E(\mathbb{F}_{q^3})$, a group whose cardinality is of order roughly $q^3$. Hence the approach illustrated in Example~\ref{ex:p=2mod3} is considerably more efficient.

From the results established above, we can construct $[\frac{q+1+2\sqrt{q}}{2}, k, \frac{q+1+2\sqrt{q}}{2}-k+1]$ MDS codes for any even integer $k \ge 3$.

\begin{corollary}
Let $q \ge 6$ be an odd square and let $k \ge 3$ be an even integer. Then there exists an MDS elliptic code $C(E, \mathcal{D}, G)$ of length $\frac{q+1+2\sqrt{q}}{2}$ and dimension $k$.
\end{corollary}

\begin{proof}
Let $E$ be a maximal elliptic curve over $\mathbb{F}_q$, that is, $N = q + 1 + 2\sqrt{q}$. Applying Theorem~\ref{thm:construction_of_degree_greater_1} together with Lemma~\ref{lem:existenceofdegree3place} to the curve $E$, the desired result follows immediately.
\end{proof}

\begin{remark}\label{remark:N/2,odd l}
In Theorem~\ref{thm:construction_of_degree_greater_1}, we construct MDS codes of length $N/2$ and dimension $k$ using a place $R$ of odd degree $l \ge 3$ satisfying $\Point([R]) = \mathcal{O}$. The proof treats even and odd $k$ separately, but the construction admits a unified description: taking a subgroup $H \subseteq E(\mathbb{F}_q)$ of index $2$, the evaluation set $\mathcal{D}$ is chosen as $H$ when $k$ is even, and as $u+H$ (with $u \notin H$) when $k$ is odd; in both cases the divisor $G$ is of the form $[R] + [P_1] + \cdots + [P_{k-l}]$ with $P_i \in H$ or $P_i \in u+H$ according to the parity of $k$. The proof that the resulting code is MDS is analogous for the two parities.

For comparison, the construction given in \cite{MR4564628} for odd $k$ uses the simpler divisor $G = k[\mathcal{O}]$ together with $\mathcal{D} = u + H$, achieving the same code parameters. {\color{blue}The construction given in \cite{wang2025nonreedsolomontypemdscodes} uses the divisor $G = k[Q]$ together with $\mathcal{D} = H$, where $Q \notin H$, also achieving the same code parameters.} Our construction replaces $k[\mathcal{O}]$ by a divisor built from a higher-degree place $R$. Without allowing $G$ to contain a place of degree greater than $1$, it is impossible to obtain MDS codes of the same length for even $k$. Our construction overcomes this barrier precisely by incorporating the higher-degree place $R$, which is essential for the even-dimensional case and also provides an alternative approach for odd $k$.
\end{remark}
\subsection{A Tight Upper Bound for The Case Where $q+1+ \lfloor 2\sqrt{q} \rfloor$ Is Even}

The construction presented in the preceding subsection yields a tight upper bound on the maximal length of MDS elliptic codes. We now state this result formally.

\begin{theorem}\label{thm:q+1+...even}
Let $q \ge 289$ and $3 \le k \le \frac{q+1-2\sqrt{q}}{10}$. Assume that $q+1+\lfloor 2\sqrt{q}\rfloor$ is even. Then
\[
\operatorname{MEC}(k,q) = \frac{q+1+\lfloor 2\sqrt{q}\rfloor}{2}.
\]
\end{theorem}

\begin{proof}
Let $E$ be an arbitrary elliptic curve defined over $\mathbb{F}_q$. From Lemma~\ref{lem:a_tight_upper_bound_before} we obtain
\[
\operatorname{MEC}(k,q,E) \le \frac{N}{2} \le \frac{q+1+\lfloor 2\sqrt{q}\rfloor}{2},
\]
and consequently
\[
\operatorname{MEC}(k,q) \le \frac{q+1+\lfloor 2\sqrt{q}\rfloor}{2}.
\]

By Lemma~\ref{lem:the_class_of_elliptic_curve}, there exists an elliptic curve $E$ with $N = q+1+\lfloor 2\sqrt{q}\rfloor$. To verify that the bound is tight, we construct MDS codes of length $\frac{q+1+\lfloor 2\sqrt{q}\rfloor}{2}$ and dimension $k$, distinguishing the parity of $k$.
Let $H$ be a subgroup of $E(\mathbb{F}_q)$ of index $2$, pick an element $u \notin H$. By Lemma~\ref{lem:existenceofdegree3place}, there exists a place $R$ of degree three. 

\begin{enumerate}
    \item \textbf{$k$ is odd.} This construction follows the method of Theorem~\ref{thm:construction_of_degree_greater_1}. 
    Set
    \[
    \mathcal{D} = u + H, \qquad G = [R] + (k-3)[\mathcal{O}].
    \]
    Then $\Point(G)=\mathcal{O} \notin \sumset_k(\mathcal{D})$, and Lemma~\ref{lem:MDS_condition} implies that the code $C(E,\mathcal{D},[R] + (k-3)[\mathcal{O}])$ is an $[\frac{q+1+\lfloor 2\sqrt{q}\rfloor}{2}, k]$ MDS code. Hence
    \[
    \operatorname{MEC}(k,q) = \frac{q+1+\lfloor 2\sqrt{q}\rfloor}{2}
    \]
    for odd $k$.

    \item \textbf{$k$ is even.} Choose a point $P \notin H$. Applying Theorem~\ref{thm:construction_of_degree_greater_1} with
    \[
    \mathcal{D} = H, \qquad G = [R] + (k-3)[P],
    \]
    yields an $[\frac{q+1+\lfloor 2\sqrt{q}\rfloor}{2}, k]$ MDS code. Thus
    \[
    \operatorname{MEC}(k,q) = \frac{q+1+\lfloor 2\sqrt{q}\rfloor}{2}
    \]
    for even $k$.
\end{enumerate}

This completes the proof.
\end{proof}

Thus, we are now in a position to answer the question posed at the outset: whether the bound $\frac{q+1+2\sqrt{q}}{2}$ in \eqref{eq:bound} is tight for even $k$ when $q$ is an odd square.

\begin{theorem}\label{thm:qoddsquare,evenk}
Let $q \ge 289$ be an odd square. Then for every even integer $k$ with $3 \le k \le (q + 1 - 2\sqrt{q})/10$, we have
\[
\operatorname{MEC}(k, q) = \frac{q+1}{2} + \sqrt{q}.
\]
\end{theorem}

\begin{proof}
This is an immediate corollary of Theorem~\ref{thm:q+1+...even}.
\end{proof}

\subsection{Examples}

   In this subsection we present several examples of MDS codes constructed from maximal elliptic curves. We employ the method introduced in Example~\ref{ex:p=2mod3} to construct a degree‑$3$ place $R$ satisfying $\Point([R]) = \mathcal{O}$.

\begin{example}
\begin{enumerate}
    \item Let $p = 17$ and consider the irreducible polynomial $f(x) = x^2 + 16x + 3$ over $\mathbb{F}_{17}$. This polynomial defines the finite field $\mathbb{F}_{17^2}$. Let $\eta$ be a root of $f(x)$ and let $E$ be the elliptic curve over $\mathbb{F}_{17^2}$ given by
    \[
    y^2 = x^3 + 1.
    \]
    The group $E(\mathbb{F}_{17^2})$ has order $324$ and is generated by the two points
    \[
    P_0 = (\eta + 5,\, 9\eta + 7), \qquad P_1 = (3\eta + 2,\, 10\eta + 7).
    \]
    One readily checks that the equation $(\eta+1)^2 = x^3 + 1$ admits no solution in $\mathbb{F}_{17^2}$. Let $R$ be the place corresponding to the ideal
    \[
    (x^3 + 14\eta + 3,\; y - \eta - 1),
    \]
    which arises from the principal divisor of $y - \eta - 1$; in other words,
    \[
    \operatorname{div}(y - \eta - 1) = [R] - 3[\mathcal{O}].
    \]
    Set $P_2 = 2P_0 = (2\eta,\, 12\eta + 16)$ and let $H = \langle P_2, P_1 \rangle$ be the subgroup generated by $P_1$ and $P_2$. Then $|H| = 162$ and $P_0 \notin H$. Take $\mathcal{D} = H$ and $G = [R] + [P_0]$. The resulting code $C(E, \mathcal{D}, G)$ is an MDS code with parameters $[162,4,159]$, which can be verified by SageMath.

    \item Let $p = 3$ and consider the irreducible polynomial $f(x) = x^6 + 2x^4 + x^2 + 2x + 2$ over $\mathbb{F}_{3}$. This polynomial defines the finite field $\mathbb{F}_{3^6}$. Let $\eta$ be a root of $f(x)$ and let $E$ be the elliptic curve over $\mathbb{F}_{3^6}$ defined by
    \[
    y^2 = x^3 + x.
    \]
    The group $E(\mathbb{F}_{3^6})$ has order $784$ and is generated by the two points
    \[
    P_0 = (2\eta^5 + 2\eta^3 + 2\eta,\, 2\eta^5 + \eta^3 + 2\eta^2 + 2\eta + 2), \]
    and 
    \[
    P_1 = (2\eta^5 + \eta^3 + \eta^2 + 2\eta + 1,\, \eta^5 + 2\eta^4 + 1).
    \]
    One readily checks that the equation $\eta^2 = x^3 + x$ admits no solution in $\mathbb{F}_{3^6}$. Let $R$ be the place corresponding to the ideal
    \[
    (x^3 + x + 2\eta^2,\; y - \eta).
    \]
    Indeed, $R$ arises from the principal divisor of $y - \eta$, namely
    \[
    \operatorname{div}(y - \eta) = [R] - 3[\mathcal{O}].
    \]
    Set $P_2 = 2P_0$.
    % = (2\eta^4 + 2\eta^3 + \eta^2 + 2\eta + 1,\, \eta^5 + 2\eta^4 + \eta^3 + 2\eta^2 + \eta + 1)$ and let $H = \langle P_2, P_1 \rangle.
    Let $H = \langle P_2, P_1 \rangle$ be the subgroup generated by $P_1$ and $P_2$. Then $|H| = 392$ and $P_0 \notin H$. Take $\mathcal{D} = H$ and $G = [R] + [P_0]$. Then SageMath shows that the resulting code $C(E, \mathcal{D}, G)$ is an MDS code with parameters $[392,4,389]$.
\end{enumerate}
\end{example}

\begin{remark}
    By Lemma \ref{lem:a_tight_upper_bound_before}, for sufficiently large odd square $q$, the bound $\frac{q+1}{2}+\sqrt{q}$ on the maximal length of MDS elliptic codes is tight for odd $k$ in a certain range. In this paper, it is proved that for even $k$, this bound is not tight if either $\mathcal{D}$ is not a coset of an index‑$2$ subgroup of $E(\mathbb{F}_q)$ or $\operatorname{Supp}(G)$ consists entirely of $\mathbb{F}_q$-rational points. In the case where $\operatorname{Supp}(G)$ consists entirely of $\mathbb{F}_q$-rational points, the tight upper bound is $\frac{q+1}{2}+\sqrt{q} - 1$.

   Conversely, if $\mathcal{D}$ is a coset of an index-2 subgroup and $\operatorname{Supp}(G)$ contains at least one place of degree greater than one, then the original bound $\frac{q+1}{2}+\sqrt{q}$ remains tight for all odd square $q$ and even $k$ within the specified range.

\end{remark}

\section{A Tight Upper Bound When $q+1+\lfloor 2\sqrt{q}\rfloor$ is odd}\label{sec:A_tight_upper_bound_for_the_case_of_characteristic_2}
This section is devoted to the complementary case where $q+1+\lfloor 2\sqrt{q}\rfloor$ is odd. This situation notably occurs when $q = 2^t$ with $t$ even, i.e., for binary fields of square cardinality, which are of particular interest in applications.
The analysis of the odd case will require us to work with elliptic curves having an odd number $N = |E(\mathbb{F}_q)|$ of rational points, and we therefore study the codes on such curves first.
With the necessary bounds for odd $N$ in place, we now proceed to determine the maximal length of MDS elliptic codes. We begin by imposing the restriction that $\operatorname{Supp}(G)$ consists entirely of $\mathbb{F}_q$-rational points. Under this restriction we obtain the tight bound $\frac{q+\lfloor 2\sqrt{q}\rfloor}{2}-1$ for even $k$ (Theorem~\ref{thm:G'and_oddq+1+}), and we then specialize it to binary fields (Corollary~\ref{cor:q=2^a,evenkandG'}). Removing the restriction on $G$ yields the exact formula $\operatorname{MEC}(k,q)=\frac{q+\lfloor 2\sqrt{q}\rfloor}{2}$ (Theorem~\ref{thm:q+1+odd}), from which we deduce the precise maximal length for characteristic $2$ (Corollary~\ref{cor:boundforcharq=2}).

\subsection{An Upper Bound on Maximal Length of MDS Codes on $E$ with Odd $N$}

We first treat the case where $N$ is odd, before proceeding to the discussion of the case where $q+1+\lfloor 2\sqrt{q}\rfloor$ is odd.

\begin{theorem}\label{thm:E(Fq)_is_odd}
Let $q$ be a prime power and let $E$ be an elliptic curve over $\mathbb{F}_q$ with $N = |E(\mathbb{F}_q)|$ odd and $N \ge 165$. Let $C(E, \mathcal{D}, G)$ be an elliptic code on $E$ of dimension $k=\deg G$ and length $|\mathcal{D}| = n$. If
\[
n \ge \frac{2}{5}N + 1 \quad \text{and} \quad 3 \le k \le n - \frac{2}{5}N + 2,
\]
then $C(E, \mathcal{D}, G)$ is not MDS.
\end{theorem}

\begin{proof}
Since $N \ge 165$ and $n \ge \frac{2}{5}N + 1$, there exists a subset $A \subseteq \mathcal{D}$ such that
\[
|A| > \max\left\{ \frac{2}{5}N,\ 12|E(\mathbb{F}_q)[2]| + 54 \right\} = \frac{2}{5}N.
\]
In particular, we may choose $A$ so that $|A| = \lfloor\frac{2}{5}N \rfloor + 1$.
By Lemma~\ref{lem:tight:first_combination_result}, it follows that $\sumset_3(A) = E(\mathbb{F}_q)$.

Now let $3 \le k \le n - \frac{2}{5}N + 2$. Select $k-3$ distinct elements $p_1, \dots, p_{k-3}$ from $\mathcal{D} \setminus A$ and denote their sum by $h$, that is,
\[
p_1 + \cdots + p_{k-3} = h \in E(\mathbb{F}_q).
\]
For an arbitrary element $g \in E(\mathbb{F}_q)$, the equality $\sumset_3(A) = E(\mathbb{F}_q)$ guarantees the existence of distinct points $a_1, a_2, a_3 \in A$ such that
\[
a_1 + a_2 + a_3 = g - h.
\]
Consequently,
\[
p_1 + \cdots + p_{k-3} + a_1 + a_2 + a_3 = g,
\]
which shows that $g \in \sumset_k(\mathcal{D})$. Since $g$ was arbitrary, we conclude that $\sumset_k(\mathcal{D}) = E(\mathbb{F}_q)$. Hence $\Point(G) \in \sumset_k(\mathcal{D})$, and by Lemma~\ref{lem:MDS_condition}, the code $C(E, \mathcal{D}, G)$ is not MDS.
\end{proof}

The following theorem gives an upper bound for the maximal length of a non‑trivial MDS elliptic code on $E$ in the case where $N$ is odd.

\begin{theorem}\label{thm:main_theorem_2}
Let $q$ be a prime power and let $E$ be an elliptic curve over $\mathbb{F}_q$ with $N = |E(\mathbb{F}_q)|$ odd and $N \ge 165$. Assume $3 \le k \le \frac{N}{10}$. Then
\[
\operatorname{MEC}(k, q, E) \le \frac{N - 1}{2} - 1.
\]
\end{theorem}

\begin{proof}
By Lemmas~\ref{lem:a_tight_upper_bound_before} and~\ref{lem:MDS_condition}, for every $k$ with $3 \le k \le N/10$ we have
\[
\operatorname{MEC}(k, q, E) \le \frac{N}{2}.
\]
Since $N$ is odd, this bound can be sharpened to
\[
\operatorname{MEC}(k, q, E) \le \frac{N - 1}{2}.
\]

To obtain the stronger result, suppose, for contradiction, that there exists an MDS code $C(E, \mathcal{D}, G)$ on $E$ with $3 \le k \le N/10$ and length
\[
|\mathcal{D}| = \frac{N - 1}{2}.
\]
We check that these parameters fulfill the hypotheses of Theorem~\ref{thm:E(Fq)_is_odd}. Indeed,
\[
|\mathcal{D}| = \frac{N - 1}{2} \ge \frac{2}{5}N + 1
\]
since $N \ge 165$, and the chain of inequalities
\[
k \le \frac{N}{10} \le |\mathcal{D}| - \frac{2}{5}N + 2
\]
also holds under the same condition on $N$. Theorem~\ref{thm:E(Fq)_is_odd} now implies that $C(E, \mathcal{D}, G)$ cannot be MDS, a contradiction. Consequently,
\[
\operatorname{MEC}(k, q, E) \le \frac{N - 1}{2} - 1,
\]
which completes the proof.
\end{proof}

\subsection{A Tight Upper Bound When $\operatorname{Supp}(G)$ Consists of $\mathbb{F}_q$-Rational Points}

As in the preceding discussion, in this subsection we restrict to the case where $\operatorname{Supp}(G)$ consists of $\mathbb{F}_q$-rational points and determine the maximal length attainable under this restriction. The following result provides the answer.
\begin{theorem}\label{thm:G'and_oddq+1+}
Let $q \ge 289$ and assume that the support of the divisor $G$ consists entirely of $\mathbb{F}_q$-rational points and that $q+1+ \lfloor 2\sqrt{q} \rfloor$ is odd. Then for every even integer $k$ with $3 \le k \le (q + 1 - 2\sqrt{q})/10$,
\begin{equation}\label{eq:MEC:q+1+isodd}
    \operatorname{MEC}^*(k,q) \le \frac{q+ \lfloor 2\sqrt{q} \rfloor}{2} - 1 .
\end{equation}
Moreover, the bound~\eqref{eq:MEC:q+1+isodd} is tight.
\end{theorem}
\begin{proof}
By Lemma~\ref{lem:a_tight_upper_bound_before}, we immediately obtain the upper bound
\[
\operatorname{MEC}(k, q) \le \frac{q+ \lfloor 2\sqrt{q} \rfloor}{2}.
\]

We now show that equality cannot hold when the support of $G$ consists entirely of $\mathbb{F}_q$-rational points. Consider any code $C(E, \mathcal{D}, G)$ satisfying the hypotheses of the theorem and having length $|\mathcal{D}| = \frac{q+ \lfloor 2\sqrt{q} \rfloor}{2}$. We prove that such a code cannot be MDS.

It suffices to treat the following two cases, as the remaining possibilities are handled analogously:
% \begin{itemize}
%     \item $N = q + \lfloor 2\sqrt{q} \rfloor$,
%     \item $N = q + \lfloor 2\sqrt{q} \rfloor + 1$.
% \end{itemize}

\noindent\textbf{Case 1:} $N = q + \lfloor 2\sqrt{q} \rfloor$. Here $N$ is even, and Corollary~\ref{cor:evenE(Fq)} yields
\[
\operatorname{MEC}(k, q, E) \le \frac{q + \lfloor 2\sqrt{q}\rfloor}{2} - 1 .
\]

\noindent\textbf{Case 2:} $N = q + \lfloor 2\sqrt{q} \rfloor + 1$. In this case $N$ is odd, and Theorem~\ref{thm:main_theorem_2} directly implies that a code with $|\mathcal{D}| = \frac{q+ \lfloor 2\sqrt{q} \rfloor}{2}$ cannot be MDS. This establishes the first claim.

It remains to verify that the bound is tight. By Lemma~\ref{lem:the_class_of_elliptic_curve}, there exists an elliptic curve $E$ over $\mathbb{F}_q$ with $N = q + \lfloor 2\sqrt{q} \rfloor$; observe that $N$ is even. Let $H \le E(\mathbb{F}_q)$ be a subgroup of index $2$, whose existence is guaranteed by Lemma~\ref{lem:the_existence_of_index_two_subgroup}. Choose an element $u \in E(\mathbb{F}_q) \setminus H$ and define
\[
\mathcal{D} := (u + H) \setminus \{u\}, \qquad G := (k+1)[\mathcal{O}] - [u].
\]
By Lemma~\ref{lem:MDS_condition}, the code $C(E, \mathcal{D}, G)$ is MDS precisely when
$
\Point(G) \notin \sumset_k(\mathcal{D})$.

Assume, for contradiction, that this condition fails. Then there exist distinct points $p_1, \dots, p_k \in \mathcal{D}$ such that
\[
\Point(G) = \sum_{i=1}^k p_i,
\]
which is equivalent to
\[
\Point((k+1)[\mathcal{O}]) = \sum_{i=1}^k p_i + u.
\]
Consequently, $\mathcal{O}$ is contained in 
$
 \sumset_{k+1}(u + H)$.
Since $u + H$ is a coset of the index‑$2$ subgroup $H$ and $k$ is even (so $k+1$ is odd), the sum of an odd number of elements from such a coset can never be $\mathcal{O}$. This contradiction shows that $C(E, \mathcal{D}, G)$ is indeed MDS, thereby proving tightness.
This completes the proof.
\end{proof}

For $q = 2^a$ with even $a$, we have the following result:
\begin{corollary}\label{cor:q=2^a,evenkandG'}
Let $q = 2^a$ with $a$ even and $q \ge 289$, and assume that the support of the divisor $G$ consists entirely of $\mathbb{F}_q$-rational points. Then for every even integer $k$ with $3 \le k \le (q + 1 - 2\sqrt{q})/10$, it holds

\begin{equation}
    \operatorname{MEC}^*(k,2^a) =2^{a-1} + 2^{a/2} - 1 .
\end{equation}

\end{corollary}
\subsection{A Tight Upper Bound When $q+1+\lfloor 2\sqrt{q}\rfloor$ is Odd}
In this subsection, We now turn to the case where no restriction is imposed on $G$.
\begin{theorem}\label{thm:q+1+odd}
Let $q \ge 289$ and $3 \le k \le \frac{q+1-2\sqrt{q}}{10}$. If $q+1+ \lfloor 2\sqrt{q} \rfloor$ is odd, then
\[
\operatorname{MEC}(k,q) = \frac{q+ \lfloor 2\sqrt{q} \rfloor}{2}.
\]
\end{theorem}

\begin{proof}
Let $E$ be an arbitrary elliptic curve defined over $\mathbb{F}_q$. By Lemma~\ref{lem:a_tight_upper_bound_before}, we have
\[
\operatorname{MEC}(k,q,E) \le \frac{N}{2} \le \frac{q+1+ \lfloor 2\sqrt{q} \rfloor}{2},
\]
and consequently
\[
\operatorname{MEC}(k,q) \le \frac{q+ \lfloor 2\sqrt{q} \rfloor}{2}.
\]

By Lemma~\ref{lem:the_class_of_elliptic_curve}, there exists an elliptic curve $E$ with $N = q + \lfloor 2\sqrt{q} \rfloor$.  Let $R$ be a place of degree three, whose existence is guaranteed by Lemma~\ref{lem:existenceofdegree3place}.  The construction of MDS codes of length $\frac{q+ \lfloor 2\sqrt{q} \rfloor}{2}$ and dimension $k$ proceeds exactly as in the proof of Theorem~\ref{thm:construction_of_degree_greater_1}: for odd $k$ we set
\[
\mathcal{D} = u + H, \qquad G = [R] + (k-3)[\mathcal{O}],
\]
(where $H$ is a subgroup of $E(\mathbb{F}_q)$ of index $2$ and $u \notin H$), while for even $k$ we set
\[
\mathcal{D} = H, \qquad G = [R] + (k-3)[u].
\]
  In either case we obtain the desired MDS code, showing that the bound is tight.  This completes the proof.
\end{proof}

As a direct corollary of Theorem~\ref{thm:q+1+odd}, we now state the bound for $q = 2^a$, where $a$ is a positive integer.

\begin{corollary}\label{cor:boundforcharq=2}
Let $q = 2^a$ for some positive integer $a$ and suppose that $q+1+\lfloor 2\sqrt{q}\rfloor$ is odd. If $q \ge 289$ and $3 \le k \le \frac{q+1-2\sqrt{q}}{10}$, then
\[
\operatorname{MEC}(k, q) = \frac{2^a + \lfloor 2\sqrt{2^a} \rfloor}{2}.
\]
In particular, when $q$ is a square (i.e., $a$ is even), this simplifies to
\[
\operatorname{MEC}(k, q) = 2^{a-1} + 2^{a/2}.
\]
\end{corollary}

\begin{remark}\label{remark:2boundforcharq=2}
Let $q = 2^a$. For completeness, we record the case where $q+1+\lfloor 2\sqrt{q}\rfloor$ is even. If $q \ge 289$ and $3 \le k \le \frac{q+1-2\sqrt{q}}{10}$, then by Theorem~\ref{thm:q+1+...even} we have
\[
\operatorname{MEC}(k, q) = \frac{2^a + 1 + \lfloor 2\sqrt{2^a} \rfloor}{2}.
\]
\end{remark}

\subsection{Example}
In this subsection, we present an explicit MDS code over a field of characteristic two. 
To construct an explicit example of an MDS code, we require an elliptic curve over a field of characteristic two whose group of rational points is cyclic. 
The following lemma is needed for this construction.
\begin{lemma}\label{lem:GCD}
\cite[Lemma 2]{MR4568185} Let $s \ge 1$ and $p > 1$ be integers. Then
\[
\gcd\left(p^{r}+1, p^{s}-1\right)
= \begin{cases}
1 & \text{if } \frac{s}{\gcd(r, s)} \text{ is odd and } p \text{ is even}, \\[4pt]
2 & \text{if } \frac{s}{\gcd(r, s)} \text{ is odd and } p \text{ is odd}, \\[4pt]
p^{\gcd(r, s)}+1 & \text{if } \frac{s}{\gcd(r, s)} \text{ is even}.
\end{cases}
\]
\end{lemma}

 From Lemma~\ref{lem:GCD}, we deduce the existence of an elliptic curve over $\mathbb{F}_q$ whose group of rational points is cyclic.

\begin{corollary}\label{cor:point:cyclic group}
Let $q = 2^{2m}$ with $m \ge 1$, and let $E$ be an elliptic curve over $\mathbb{F}_q$ such that $N = q + 2\sqrt{q}$. If $m$ is odd, or if $m$ is even and $m \equiv 0$ or $2 \pmod{6}$, then the group of rational points $E(\mathbb{F}_q)$ is cyclic.
\end{corollary}

\begin{proof}
Suppose that $ N $ admits the prime factorization 
\[ N = \prod_{ l } l^{h_l} ,  \]
where each index $ l $ is a prime.
By Lemma~\ref{lem:the_structure_of_group}, every possible group structure for $E(\mathbb{F}_q)$ of this order is of the form
\[
\mathbb{Z} / p^{h_p} \mathbb{Z} \times \prod_{l \neq p} \left( \mathbb{Z} / l^{a_l} \mathbb{Z} \times \mathbb{Z} / l^{h_l - a_l} \mathbb{Z} \right).
\]

Set $t = N - q - 1$. From the given hypotheses we have
\[
\gcd(t, p) = \gcd(2^{m+1} - 1, 2) = 1.
\]

Applying Lemma~\ref{lem:the_structure_of_group}, we get
\begin{align*}
\gcd(N, q-1)
&= \gcd(2^{2m} + 2^{m+1}, 2^{2m} - 1) \\
&= \gcd(2^{m-1} + 1, 2^{2m} - 1) .
\end{align*}
It follows from Lemma~\ref{lem:GCD} that
\[\gcd(N, q-1) = \begin{cases}
1, & \text{if } \dfrac{2m}{\gcd(m-1, 2m)} \text{ is odd}, \\[6pt]
2^{\gcd(m-1, 2m)} + 1, & \text{if } \dfrac{2m}{\gcd(m-1, 2m)} \text{ is even}.
\end{cases}
\]

The parity of $\frac{2m}{\gcd(m-1, 2m)}$ is determined as follows:
\begin{enumerate}
    \item If $m$ is odd, then $2 \mid \gcd(m-1, 2m)$, and consequently $\frac{2m}{\gcd(m-1, 2m)}$ is odd.
    \item If $m$ is even, then $\gcd(m-1, 2m) = \gcd(m-1, m) = 1$, so that $\frac{2m}{\gcd(m-1, 2m)}$ is even.
\end{enumerate}
Therefore,
\begin{align*}
\gcd(N, q-1)
% &=
% \begin{cases}
% 1, & \text{if } m \text{ is odd}, \\
% 2^{\gcd(m-1, 2m)} + 1, & \text{if } m \text{ is even}
% \end{cases} \\
&=
\begin{cases}
1, & \text{if } m \text{ is odd}, \\
3, & \text{if } m \text{ is even}.
\end{cases}
\end{align*}

Now let $l$ be a prime divisor of $N$ with $l \neq 2$, and denote by $v_l(n)$ the $l$-adic valuation of an integer $n$.

\begin{itemize}
    \item[(i)] If $m$ is odd, then $\gcd(N, q-1) = 1$, whence $v_l(q-1) = 0$. By Lemma~\ref{lem:the_structure_of_group} (b) we obtain
    \[
    0 \le a_l \le \min\{0, \lfloor h_l/2 \rfloor\},
    \]
    which forces $a_l = 0$. The group structure thus reduces to
    \[
    \mathbb{Z} / p^{h_p} \mathbb{Z} \times \prod_{l \neq p} \bigl( \mathbb{Z} / l^{h_l} \mathbb{Z} \bigr),
    \]
    and is therefore cyclic.

    \item[(ii)] If $m \equiv 0 \pmod{6}$ or $m \equiv 2 \pmod{6}$, then $\gcd(N, q-1) = 3$. For any prime $l \neq 3$ we again have $v_l(q-1) = 0$, leading to $a_l = 0$ by the same argument.

    It remains to analyze the prime $l = 3$. Since $2$ and $3$ are coprime, we only need to examine the exponent of $3$ in $2^{m-1}+1$. From the previous calculations this exponent is at least $1$. If $m \equiv 0 \pmod{6}$, write $m = 6k$ with $k \ge 1$. Then
    \[
    2^{m-1}+1 \equiv 2^{6k-1}+1 \equiv 2^5+1 = 33 \equiv 6 \pmod{9},
    \]
    showing that the exponent of $3$ is exactly $1$; that is, $h_3 = 1$. By Lemma~\ref{lem:the_structure_of_group} (b) we have
    \[
    0 \le a_3 \le \min\{v_3(q-1), \lfloor 1/2 \rfloor\} = 0,
    \]
    hence $a_3 = 0$. The case $m \equiv 2 \pmod{6}$ is handled similarly. Consequently, the group is again cyclic.
\end{itemize}

This completes the proof.
\end{proof}

From the elliptic curve discussed above, one can construct MDS codes that attain the upper bound. The following example illustrates this construction.

\begin{example}
Choose an irreducible polynomial 
\[ f(x) = x^{10} + x^6 + x^5 + x^3 + x^2 + x + 1\]
over $\mathbb{F}_2$.
We have $ \mathbb{F}_{2^{10}} \cong \mathbb{F}_q (\eta) $, where $\eta$ is a root of $f(x)$. Define the elliptic curve $E$ over $\mathbb{F}_{2^{10}}$ by
\[
y^2 + xy = x^3 + (\eta^8 + \eta^6 + \eta^2).
\]
The group $E(\mathbb{F}_{2^{10}})$ has order $1088$. By Corollary \ref{cor:point:cyclic group}, the group structure is cyclic. A typical generator $P_0$ of $E(\mathbb{F}_{2^{10}})$ is given by
\[
(\eta^7 + \eta^6 + \eta^4 + \eta^2 ,\, \eta^4 + \eta^3).
\]
Let $H = \langle 2 P_0 \rangle$ and $\mathcal{D} = P_0 + H$. Taking $G = 3[\mathcal{O}]$, the resulting code $C(E, \mathcal{D}, G)$ is an $[544,3,542]$ MDS code, verified by SageMath.
\end{example}

\section{Conclusion}\label{sec:conclusion}
In this work, we have investigated the maximal length of non‑trivial MDS elliptic codes, with particular attention to the remaining open cases concerning even dimension $k$, non‑square $q$, and characteristic‑$2$ fields. Our main contributions are twofold.

First, we showed that under the conventional assumption that the divisor $G$ is supported entirely on $\mathbb{F}_q$-rational points, the conjectured bound $\frac{q+1}{2}+\sqrt{q}$ is not attainable for even $k$ when $q$ is an odd square. Specifically, for $q \ge 289$ and $3 \le k \le (q+1-2\sqrt{q})/10$, we proved that no MDS elliptic code of length $\frac{q+1}{2}+\sqrt{q}$ and even dimension $k$ can exist under this restriction. Instead, the exact maximal length in this setting is $\frac{q+1}{2}+\sqrt{q}-1$. To circumvent this limitation, we introduced a new construction that allows the support of $G$ to include places of degree strictly greater than $1$. By establishing the existence of a suitable place of degree $3$ on the underlying elliptic curve, we provided explicit MDS codes of length $\frac{q+1}{2}+\sqrt{q}$ for even $k$, thereby confirming that the bound is tight in the unrestricted setting.

Second, we extended the analysis to the case where $q+1+\lfloor 2\sqrt{q}\rfloor$ is odd, a scenario that naturally encompasses the case where $q$ is an even square. We derived the tight upper bound $\operatorname{MEC}(k,q) = \frac{q+\lfloor 2\sqrt{q}\rfloor}{2}$ for all admissible $k$ and, moreover, determined the precise maximal length when the support of $G$ is restricted to $\mathbb{F}_q$-rational points. As a notable consequence, we obtain an exact formula for the maximal length of nontrivial MDS elliptic codes over fields of characteristic $2$: for $q = 2^a \ge 289$ and $3 \le k \le (q+1-2\sqrt{q})/10$,
\[
\operatorname{MEC}(k,q) = 
\begin{cases}
\dfrac{2^a + \lfloor 2\sqrt{2^a}\rfloor}{2}, & \text{if } 2^a+1+\lfloor 2\sqrt{2^a}\rfloor \text{ is odd},\\[8pt]
\dfrac{2^a + 1 + \lfloor 2\sqrt{2^a}\rfloor}{2}, & \text{if } 2^a+1+\lfloor 2\sqrt{2^a}\rfloor \text{ is even}.
\end{cases}
\]
In the odd case, when $q$ is a square (i.e., $a$ even), this simplifies to $2^{a-1} + \sqrt{2^a}$.

The constructions presented in this paper rely on the existence of maximal elliptic curves and specific degree-$3$ places, which we have shown to be available under mild conditions on the field size and characteristic. Our results provide a complete resolution of the tightness question for the bound $\frac{q+1}{2}+\sqrt{q}$ in the even-dimensional case and offer a unified framework for determining $\operatorname{MEC}(k,q)$ across all relevant parity combinations of $q$ and $k$.

A natural direction for future work is to investigate whether similar improvements can be achieved for dimensions $k$ outside the range $3 \le k \le (q+1-2\sqrt{q})/10$. It would also be of interest to explore whether the techniques developed here--in particular, the use of higher-degree places in the support of $G$--can be adapted to other families of algebraic geometry codes, or to extend the analysis to curves of higher genus.

\section*{Acknowledgement}

This work is supported by the National Natural Science Foundation of China (No. 12441107),   Guangdong Basic and Applied Basic Research Foundation (No. \seqsplit{2025A1515011764}), and the National Key Research and Development Program of
China (No.\seqsplit{2025YFA1017100}).

\bibliographystyle{IEEEtran}
\bibliography{reference}

\end{document}